\documentclass[11pt]{article}

\usepackage[margin=1in]{geometry}
\usepackage{amsmath,amssymb,amsthm}
\usepackage{graphicx}
\usepackage{hyperref}
\usepackage[numbers,sort&compress]{natbib}
\usepackage{microtype}

\newcommand{\orcid}[1]{}
\newcommand{\isAPAStyle}[2]{}
\newcommand{\isChicagoStyle}[2]{}
\newcommand{\pubyear}[1]{}
\newcommand{\copyrightyear}[1]{}
\newcommand{\hreflink}[1]{}
\newcommand{\externaleditor}[1]{}
\newcommand{\firstpage}[1]{}
\newcommand{\articlenumber}[1]{}
\newcommand{\issuenum}[1]{}
\newcommand{\pubvolume}[1]{}

\makeatletter
\@ifundefined{Theorem}{\newtheorem{Theorem}{Theorem}}{}
\@ifundefined{Lemma}{}{}
\@ifundefined{Proposition}{}{}
\@ifundefined{Corollary}{}{}
\theoremstyle{definition}
\@ifundefined{Definition}{\newtheorem{Definition}{Definition}}{}
\@ifundefined{Remark}{}{}
\@ifundefined{Example}{}{}
\makeatother

\providecommand{\keywords}[1]{\bigskip\noindent\textbf{Keywords: } #1}
\usepackage{bm}
\usepackage{bm}

\newcommand{\Tr}{\mathrm{Tr}}

\newcommand{\tr}{\mathrm{tr}}
\newcommand{\calX}{\mathcal{X}}
\newcommand{\bbR}{\mathbb{R}}

\usepackage{authblk}

\begin{document}

\title{Hybrid Cram\'er-Rao bound for Quantum Bayes-Point Estimation with Nuisance Parameters}

\author[1,*]{Jianchao Zhang}
\author[1,2]{Jun Suzuki}
\affil[1]{Graduate School of Informatics and Engineering, The University of Electro-Communications (UEC), Tokyo, Japan}
\affil[2]{Institute for Advanced Science, The University of Electro-Communications (UEC), Tokyo, Japan}
\affil[*]{Correspondence: c2141016@edu.cc.uec.ac.jp}

\date{}
\maketitle

\abstract{
We develop a hybrid framework for quantum parameter estimation in the presence of nuisance parameters. In this Bayes-point scheme, the parameters of interest are treated as fixed non-random parameters while nuisance parameters are integrated out with respect to a prior (random parameters). Within this setting, we introduce the hybrid partial quantum Fisher information matrix (hpQFIM), defined by prior-averaging the nuisance block of the QFIM and taking a Schur complement, and derive a corresponding Cram\'er-Rao-type lower bound on the hybrid risk. We establish structural properties of the hpQFIM, including inequalities that bracket it between computationally tractable surrogates, as well as limiting behaviors under extreme priors. Operationally, the hybrid approach improves over pure point estimation since the optimal measurement for the parameters of interest depends only on the prior distribution of the nuisance, rather than on its unknown value. We illustrate the framework with analytically solvable qubit models and numerical examples, clarifying how partial prior information on nuisance variables can be systematically exploited in quantum metrology.
}





\section{Introduction}
Quantum metrology and quantum sensing have matured into rigorous frameworks for quantum-limited precision measurement, with rapid theoretical and experimental progress in recent years \cite{degen17,giovannetti2011,pezze2018quantum,Paris2009,demkowicz2012elusive,ying2022critical}.
In many such tasks, the parameter vector naturally separates into parameters of interest, which encode the physical quantity we ultimately care about, and nuisance parameters, which affect the data but are not themselves the target \cite{suzuki2020nuisance,suzuki2020quantum}. Typical nuisances include optical loss and detector inefficiency in interferometry \cite{demkowicz2015quantum}, unknown phase or polarization offsets due to misalignment, dephasing and amplitude-damping rates in spectroscopy \cite{degen17}, background counts in imaging \cite{moreau2019imaging}, or slow drifts in local oscillators for frequency standards \cite{santarelli2002frequency}. Treating interest and nuisance on equal footing can blur the operational goal and can also reduce statistical efficiency \cite{suzuki2020quantum}: measurement settings that are ideal for learning the nuisance may be suboptimal for the scientific quantity of interest \cite{ragy2016compatibility}. 

Quantum estimation provides the decision-theoretic backbone for metrology and sensing by linking experimental design (choice of measurement) to achievable precision limits. Quantum parameter estimation admits both frequentist and Bayesian formulations \cite{barndorff2000fisher,li2018frequentist}. In point (frequentist) formulations, locally optimal measurements often depend on the unknown true parameter; in multi-parameter models, incompatibility between observables can prevent simultaneous attainment of single-parameter limits \cite{braunstein1994statistical}. In fully Bayesian formulations, performance is optimized on average with respect to a prior, which improves robustness and ease of implementation but may reduce local efficiency when the prior is diffuse or misspecified \cite{gill1995applications}. In practice-atomic clocks \cite{ludlow2015optical}, magnetometry \cite{budker2007optical}, optical phase tracking \cite{yonezawa2012quantum}, nanoscale imaging \cite{barry2020sensitivity}. We often have partial prior information about nuisance parameters from routine characterization, while the scientific parameters of interest still demand local, high-resolution treatment \cite{suzuki2020quantum}. This operational asymmetry motivates a hybrid approach.

\subsection{Contributions of this paper}
\begin{enumerate}
  \item \textbf{Framework and risk.} We formalize a hybrid estimation framework that treats parameters of interest as fixed non-random parameters while incorporating nuisance parameters through a prior distribution, i.e., random parameters obeying the distribution. We introduce a hybrid mean squared error (MSE) and hybrid risk as the objective to minimize. (Definition~\ref{def:hybrid-risk}).
  \item \textbf{Hybrid CR-type lower bound.} We prove a Cram\'er-Rao-type (CR-type) inequality in the hybrid setting, identifying the hybrid partial quantum Fisher information matrix (hpQFIM) as a fundamental lower bound on the hybrid risk for the interest parameters under admissible measurements and estimators (Theorem~\ref{thm:qhybridCR}).
  \item \textbf{Two-sided approximations and ordering relations.} We establish computable upper and lower approximations for the hpQFIM (Theorem~\ref{thm:hybridlbub}).
\end{enumerate}

\subsection{Short summary of point estimation and Bayesian estimation} 

\subsubsection{Point estimation in quantum models}\label{subsec:intro-point-nomath}
In point estimation, one fixes an unknown parameter value and seeks measurements and estimators that are locally efficient around that point. Classical CR-type guarantees relate the achievable MSE to information carried by the measurement outcomes, and in quantum settings the choice of measurement becomes part of the optimization itself \cite{helstrom1976quantum,hayashi-book}. In multi-parameter models, jointly optimal measurements can be hindered by incompatibility among observables \cite{conlon2023multiparameter}; Holevo-type criteria capture the best trade-offs permitted by quantum mechanics \cite{holevobook,hayashi2008asymptotic}. A practical limitation is that locally optimal measurements typically depend on the unknown parameter, so adaptive or two-stage strategies are often used to first localize and then refine \cite{gill2000,hayashi2008asymptotic,fujiwara06consistency}.

\subsubsection{Bayesian estimation and prior-averaged optimality}\label{subsec:intro-bayes-nomath}
Bayesian estimation evaluates performance on average with respect to a prior over parameters. The optimal measurement-estimator pair minimizes the Bayes risk defined by this prior, and fundamental lower bounds- such as van Trees-type inequalities and their quantum analogues- link Bayes risk to prior-averaged information quantities \cite{van2004detection,gill1995applications}. In quantum settings, several quantum versions of these classical bounds have been proposed \cite{personick1971application,rubio2019quantum,demkowicz2020multi}. 
Recent Bayesian logarithmic-derivative (LD)-type bounds provide convenient and often tighter computable benchmarks in finite-copy regimes \cite{zhang2024bayesian}. In practice, the optimal Bayesian measurement depends on the prior rather than the unknown true value; this reduces design complexity when only distributional knowledge is available or when adaptive localization is expensive.

\subsubsection{Motivation for a hybrid framework}
Point estimation excels at local precision but can require rapid localization and may face incompatibility in the multi-parameter regime. Full Bayesian estimation is robust and measurement-friendly, yet local sharpness can be diluted under diffuse or misspecified priors. Many metrological scenarios lie between these extremes: we often possess actionable prior information about nuisance components while still aiming for the best local performance on the parameters of interest. This motivates a hybrid approach that uses prior averaging for nuisance parameters to gain robustness and implementability, while preserving point-wise efficiency for the parameters of interest. This perspective resonates with the hybrid CR lower bound in classical signal processing~\cite{messer2006hybrid,fortunati2017performance}. In this work we extend these ideas to quantum estimation; formal definitions and bounds shall be presented in the paper.

The remainder of this paper is organized as follows. 
In Section~2, we formulate the hybrid estimation framework, define the hpQFIM, and derive the associated hybrid CR-type lower bound together with computable inequalities. 
In Section~3, we present numerical case studies on noisy qubit models to illustrate the behavior of the proposed hybrid bound under different nuisance structures. 
In Section~4, we analyze an analytically solvable qubit example where directional parameters are estimated in the presence of a radial nuisance, highlighting how the hybrid formulation connects to the state. 
Section~5 concludes the paper and discusses open directions.
Detailed proofs of the main theorems are provided in Appendices~A and~B.

\section{Hybrid Framework} \label{sec:hybrid_framework} 
In this section, we provide the hybrid framework and two main theorems. The proof of the theorem is written in appendix.
\subsection{Setting and notation}
The parametric model is $ \left\{ {\rho_{\theta_I,\theta_N}} \mid \theta_I \in \Theta_I \subset \mathbb{R}^{d_I}, ~ \theta_N \in \Theta_N \subset \mathbb{R}^{d_N}  \right\}$ with the state \(\rho_{\theta_I, \theta_N} \) on a finite dimensional Hilbert space \( \mathcal H\). Let the parameter vector be partitioned as \(\theta = (\theta_I,\theta_N)\) such that
\begin{align}
    \theta=(\theta_1,\theta_2,\cdots,\theta_{d_I},\theta_{d_I+1},\cdots,\theta_{d_I + d_N}) \in 
    \Theta_I\times\Theta_N\subset \bbR^{d_I + d_N},
\end{align}
where \(d_I\) and \(d_N\) represents the numbers of parameter of interest and nuisance parameters. Thus we use the \(\theta_I\) and \(\theta_N\) to present the vector of corresponding parameters. The positive operator-valued measure (POVM) is
$\Pi=\{ \Pi_x \mid x\in \mathcal{X} \}$ with outcome
$x\sim p(x \,|\, \theta_I,\theta_N)=\tr[\rho_{\theta_I,\theta_N}\Pi_x]$
such that $\Pi_x \succeq 0$ and
\begin{align}
\int_{\mathcal{X}} \Pi_x\,\mathrm dx=I, 
\end{align}
where the integral is taken over $\mathcal{X}$. In the purely discrete case,  the integral is understood as a sum,
i.e.,\ $\sum_{x\in\mathcal{X}}\Pi_x=I$. We use the integral notation for uniformity and all statements apply equally to discrete outcome spaces.
The estimator $\hat\theta_I: \calX \mapsto \Theta_I $ be locally unbiased at \(\theta_I\) for any \(\theta_N\) which means
\begin{equation}\label{eq:localu}
    \begin{split}
    \mathbb{E}_{x|\theta_I,\theta_N} [\hat\theta_{I,i}(x)] &= \theta_{I,i} ~\text{  for all }i, \\
    \frac{\partial}{\partial \theta_{I,i}} \mathbb{E}_{x|\theta_I,\theta_N} [\hat\theta_{I,j}(x)] &= \delta_{i,j} ~\text{   for all } i,j,    
    \end{split}
\end{equation}
where \(\theta_{I,i}\) denotes the \(i\)-th parameter of interest and the expectation is 
\begin{align}
    \mathbb{E}_{x|\theta_I,\theta_N}[f(x)] &:= \int_\mathcal{X}  f(x)  p(x | \theta_I,\theta_N) dx.
\end{align}
In this research, the random variable is denoted in small \(x\).

\begin{Definition}[Hybrid MSE]\label{def:hybrid-risk} Given a prior density $\pi$ on $\Theta_N$, the (matrix-valued) hybrid MSE is 
\begin{align} 
V_{\theta_I,\pi}(\Pi,\hat\theta_I) :=\mathbb{E}_{\pi} \mathbb{E}_{x|\theta_I,\theta_N} \big[(\hat\theta_I(x)-\theta_I)(\hat\theta_I(x)-\theta_I)^{\top} \big] \in \bbR^{d_I \times d_I} ,
\end{align}
where the expectations are 
\begin{align}
    \mathbb{E}_{x|\theta_I,\theta_N}[f(x)] &:= \int_\mathcal{X}  f(x)  p(x | \theta_I,\theta_N) dx,\\
    \mathbb{E}_{\pi}[g(\theta_N)] &:= \int_{\Theta_N}  g(\theta_N) ~\pi(\theta_N) \, d\theta_N.
\end{align}
\end{Definition}

\paragraph{The main objective}
The original main aim of quantum estimation is to finding the following quantity.
\begin{align}
    &\min_{\Pi,\hat\theta_I} V_{\theta_I,\pi}(\Pi,\hat\theta_I), \\
    &\text{s.t. } \hat \theta_I \text{ is locally unbiased at } \theta_I \text{ for any } \theta_N .
\end{align}
This minimization is taken over \( (\Pi,\hat\theta_I )\) which is called quantum decision. However, this minimization is not always possible because it is a matrix. Thus the main objective is to minimize the scalar hybrid risk for a weight matrix, $W \succ 0$ on parameters of interest  $\mathbb{R}^{d_I}$. The hybrid risk is defined as follows.
\begin{Definition}[Hybrid risk] 
\begin{align}
\mathcal{R}_{\theta_I,\pi} (\Pi,\hat\theta_I\mid W)
:=\Tr\big[W \cdot V_{\theta_I,\pi}(\Pi,\hat\theta_I)\big].
\end{align} \end{Definition}
The prior \(\pi(\theta_N)\) is assumed to be twice continuously differentiable in \(\theta_N\) and decays sufficiently fast at the boundary (or at infinity) so that all boundary terms vanish in first-order integration by parts.

\subsection{Quantum information blocks and the hpQFIM}
Let $J$ denote the symmetric logarithmic derivative (SLD) quantum Fisher information matrix (QFIM) and write its block form
\begin{align}
J(\theta_I,\theta_N)
=\begin{pmatrix}
J_{II}(\theta_I,\theta_N) &
J_{IN}(\theta_I,\theta_N) \\
J_{NI}(\theta_I,\theta_N) &
J_{NN}(\theta_I,\theta_N)
\end{pmatrix} \in \bbR^{ (d_I + d_N) \times (d_I + d_N) } .
\end{align}
In this research, SLD QFIM is involved since this is the most widely used quantum score operator due to its symmetric and Hermitian properties \cite{helstrom1969quantum}. However, the result is free to extended in right logarithmic derivative (RLD) QFIM \cite{yuan1970optimal}. 
In what follows, we omit the explicit \((\theta_I,\theta_N)\) dependence in QFIM blocks when clear from context.
We use the partial QFIM for the parameters of interest, defined as the Schur complement of the nuisance block,
\begin{align}\label{eq:pqfim}
  J_{I|N} \;:=\; J_{II} \;-\; J_{IN}\, J_{NN}^{-1}\, J_{NI},
\end{align}
which captures the information on \(\theta_I\) after optimally projecting out the effect of the nuisance
parameters \(\theta_N\) and is the appropriate quantity entering the CR-type bound for \(\theta_I\) \cite{kumon1984estimation,amari1988estimation}.

Let $J_\pi$ be the classical Fisher information matrix of the prior $\pi$ on $\theta_N$ of which the \(a,b\)-th entry is
\begin{align}
    J_{\pi,ab} := \int_{\Theta_N} \left( \frac{\partial}{\partial \theta_{N,a}} \log \pi(\theta_N) \right)    \left( \frac{\partial}{\partial \theta_{N,b}} \log \pi(\theta_N) \right)  \pi(\theta_N)\, d\theta_N.
\end{align}
We reserve the indices \(i,j\) for components of \(\theta_I\) (written \(\theta_{I,i}, \theta_{I,j}\)) and the indices \(a,b\) for components of \(\theta_N\) (written \(\theta_{N,a}, \theta_{N,b}\)). We define the \emph{hpQFIM} by prior-averaging the blocks over $\pi$ and taking the Schur complement:

\begin{Definition}[Hybrid Partial Quantum Fisher Information Matrix (hpQFIM)]\label{def:hybrid partial QFIM}
\begin{align}
J^{(\pi)}_{I\mid N}(\theta_I)
:=\ \mathbb{E}_{\pi}\big[J_{II}\big]
\;-\;
\mathbb{E}_{\pi}\big[J_{IN}\big]\,
\Big(\,\mathbb{E}_{\pi}\big[J_{NN}\big]+J_\pi\Big)^{-1}
\mathbb{E}_{\pi}\big[J_{NI}\big].
\end{align}
\end{Definition}

\noindent
Intuitively, $J^{(\pi)}_{I\mid N}$ aggregates (i.e., prior-averages out) nuisance information and quantifies how much information for $\theta_I$ remains after accounting for the prior on $\theta_N$.

\subsection{Hybrid CR-type lower bound}

We state the main result of the paper, which is proven by using the covariance inequality. 
\begin{Theorem}[Quantum hybrid CR-type lower bound]\label{thm:qhybridCR}
For any POVM $\Pi$ and any locally unbiased estimator for the parameters of interest $\hat\theta_I$, the matrix inequality
\begin{align}
V_{\theta_I,\pi}(\Pi,\hat\theta_I)\ \succeq\ \big(J^{(\pi)}_{I\mid N}(\theta_I)\big)^{-1},
\end{align}
holds. Hence, for any weight matrix $W\succ 0$, the hybrid risk is bounded as
\begin{align}
\mathcal{R}_{\theta_I,\pi}(\Pi,\hat\theta_I \mid W) \ \ge\ \Tr\Big[\,W\,\big(J^{(\pi)}_{I\mid N}(\theta_I)\big)^{-1}\Big].
\end{align}
\end{Theorem}
\begin{proof}[Proof]
Details in appendix \ref{app:Quanhybridcrlowerbound}.
\end{proof}

\paragraph{Motivation for Two-sided approximations.}
The hpQFIM $J^{(\pi)}_{I|N}$ is the central information quantity in our framework, but it is not always the most convenient object to evaluate or compare across models and priors. In quantum point estimation with nuisance parameters, related results bound the Schur-complement-type information between an averaging-after-inverse quantity and a simpler interest-block average. Such bounds serve two purposes: (i) they provide computationally convenient approximations when the exact partial information is hard to obtain, and (ii) they characterize how much precision can be gained or lost due to nuisance parameter and prior uncertainty. Motivated by this practice, we establish analogous two-sided approximations for the hybrid setting.

\begin{Theorem}[Lower and upper approximations for the hpQFIM]\label{thm:hybridlbub}
For the hpQFIM $J^{(\pi)}_{I|N}$, the following inequalities hold:
\begin{align}
\mathbb{E}_\pi[J_{II}(\theta_I,\theta_N)]
\ \succeq\ 
J^{(\pi)}_{I|N}(\theta_I)
\ \succeq\
\mathbb{E}_\pi \left[J_{I|N}(\theta_I,\theta_N)\right] .
\end{align}
\end{Theorem}
\begin{proof}[Proof]
Details in appendix \ref{app:lowerandupper}.
\end{proof}

We have two remarks on the theorem. 
\begin{itemize}
\item Computational surrogates. For our model, the partial information $J_{I|N}$ has a closed form at each nuisance sample, so the right bound
$\mathbb{E}_\pi[J_{I|N}(\theta_I,\theta_N) ]$ is evaluated by \emph{averaging closed-form matrices} and then inverting once.
The left bound $\mathbb{E}_\pi[J_{II}(\theta_I,\theta_N)]$ is even simpler: average once and invert once.
\item 
By contrast, the hybrid quantity requires the term $\big(\mathbb{E}_{\pi}[J_{NN}]+J_\pi\big)^{-1}$ does not admit a closed form in general (it depends on the prior $\pi$). Consequently, it typically has to be computed numerically and repeatedly (e.g., across prior hyperparameters), making this middle term the computational bottleneck. This is precisely why the left/right bounds are useful: they bracket $J^{(\pi)}_{I|N}$ while avoiding repeated inner inversions.
\end{itemize}

\section{Examples (noisy qubit metrology)}\label{sec:results}

This section complements the hybrid framework in Section~\ref{sec:hybrid_framework} by reporting numerical comparisons on qubit models in Bloch sphere parameters. An analytically solvable model will be presented in the next section. 

\subsection{Numerical comparison on qubit models}\label{subsec:numerical-qubit}
We consider single-qubit models with two parameters $\theta=(\theta_I,\theta_N)$ and priors $\pi$ on nuisance.
For each model we report:
(i) the lower bound for hybrid risk $\Tr \big[W\,(J^{(\pi)}_{I\mid N}(\theta_I))^{-1}\big]$ (Definition~\ref{def:hybrid partial QFIM}); 
(ii) the lower bound and upper bound for the hpQFIM (Theorem~\ref{thm:hybridlbub})
Unless stated otherwise, weights use $W=I$.

The model is presented in Bloch sphere parameter as 
\begin{align}
    \rho_{\theta_I,\theta_N} = \frac{1}{2} \left( I + \mathbf{s}(\theta_I,\theta_N \mid r,\phi) \cdot \bm{\sigma}  \right),
\end{align}
where \(\bm{\sigma}=(\sigma_x,\sigma_y,\sigma_z)\) is the vector of Pauli matrices. We study three representative cases:
\begin{itemize}
  \item \textbf{(1) Phase with extra rotation:} interest $(\theta_I)$, nuisance $(\theta_N)$; we sweep prior concentration to illustrate the gap predicted.
  \begin{equation}
      \mathbf{s} (\theta_I,\theta_N \mid r,\phi) = (r \sin \phi \cos (\theta_I + \theta_N) , r \sin \phi \sin (\theta_I +\theta_N), r \cos \phi),
  \end{equation}
  with parameters \(\theta_I \in [0, 2\pi),~ \theta_N \in [0, 2\pi)\) and fixed values \( r\in(0,1),~\phi\in [0,2\pi).\) This model is impossible to estimate in point estimation but is available in this hybrid framework.
  The result is illustrated in Figure~\ref{fig:extrarotation}.

  \item \textbf{(2) Additional-sine model:} interest $(\theta_I)$, nuisance $(\theta_N)$ with cross-coupling; priors on $\theta_N$ with varying concentration.
  \begin{equation}
      \mathbf{s} (\theta_I,\theta_N \mid r,\phi) = (r \sin \phi \cos \theta_I, r \sin \phi \sin (\theta_I + \theta_N), r \cos \phi),
  \end{equation}
  with parameters \(\theta_I \in [0, 2\pi),~ \theta_N \in [0, 2\pi) \) and fixed values \( r\in(0,1),~\phi\in [0,2\pi)\) with a constraint \(r^2 ( \sin^2 \phi +1) \leq 1\).
  This model can be interpreted as a toy abstraction of situations where an additional phase affects only one channel (e.g., a single arm of an interferometer or one polarization component), thereby breaking the usual rotational symmetry.
  The result is illustrated in Figure~\ref{fig:additionalsine}.
  
 \item \textbf{(3) Anisotropic shrinking:} dissipative channel with axis-dependent contraction; we isolate interest while averaging nuisance shrinkage.
  \begin{equation}
      \mathbf{s} (\theta_I,  \theta_N \mid r,\phi) = (r \sin \phi \cos \theta_I,  r \, \theta_N \sin \phi  \sin \theta_I, r \cos \phi  ),
  \end{equation}
  with parameters \(\theta_I \in [0, 2\pi),~ \theta_N \in (0, 1] \) and fixed values \( r\in(0,1),~ \phi\in [0,2\pi)\).
  The result is illustrated in  Figure~\ref{fig:anisotropic}.
\end{itemize}

\paragraph{Steps for numerical analysis.}
We repeat the following steps to analyze each model. 
The results will be given in the next subsections. 
\begin{enumerate}
  \item Compute the QFIM $J(\theta_I,\theta_N)$ and the  partial QFIM $J_{I\mid N}(\theta_I,\theta_N)$ via Equation~\ref{eq:pqfim}.
  \item Let \(\pi\) be uniform distribution in the domain of nuisance parameter (non-informative prior), compute the analytical form of hpQFIM (Definition~\ref{def:hybrid partial QFIM}) and its corresponding lower and upper approximations (Theorem~\ref{thm:hybridlbub}.) One may notice that in Theorem~\ref{thm:qhybridCR}, the lower bound of hybrid risk is the inverse of the hpQFIM. Thus, in this section, we demonstrate the lower and upper approximations in the inverse form as
  \begin{align}
\left(\mathbb{E}_\pi \left[J_{I|N}(\theta_I,\theta_N)\right] \right)^{-1}
\ \succeq\ 
\big(J^{(\pi)}_{I|N}(\theta_I)\big)^{-1}
\ \succeq\ \big(\mathbb{E}_\pi[J_{II}(\theta_I,\theta_N)]\big)^{-1}.
\end{align}
  \item Derive the values for former quantities for grid points (fifty points in each figure) of parameter of interest in its range and truncate the point on the boundary.
\end{enumerate}

\subsection{Extra rotation model}
For the model
\[
\mathbf{s}(\theta_I,\theta_N\mid r,\phi)
=\big(r\sin\phi\cos(\theta_I+\theta_N),\ r\sin\phi\sin(\theta_I+\theta_N),\ r\cos\phi\big),
\]
the score directions satisfy $\partial_{\theta_I}\mathbf s=\partial_{\theta_N}\mathbf s$. Hence the single-qubit QFIM blocks are constants
\[
J_{II}=J_{NN}=J_{IN}=r^2\sin^2\phi,
\]
independent of the parameters $(\theta_I,\theta_N)$, and the likelihood-only Schur complement is identically zero:
$J_{I|N}=J_{II}-J_{IN}^2/J_{NN}=0$. The result is shown in Figure~\ref{fig:extrarotation}.
\begin{figure}[!htbp]
    \centering
    \includegraphics[width=1.0\linewidth]{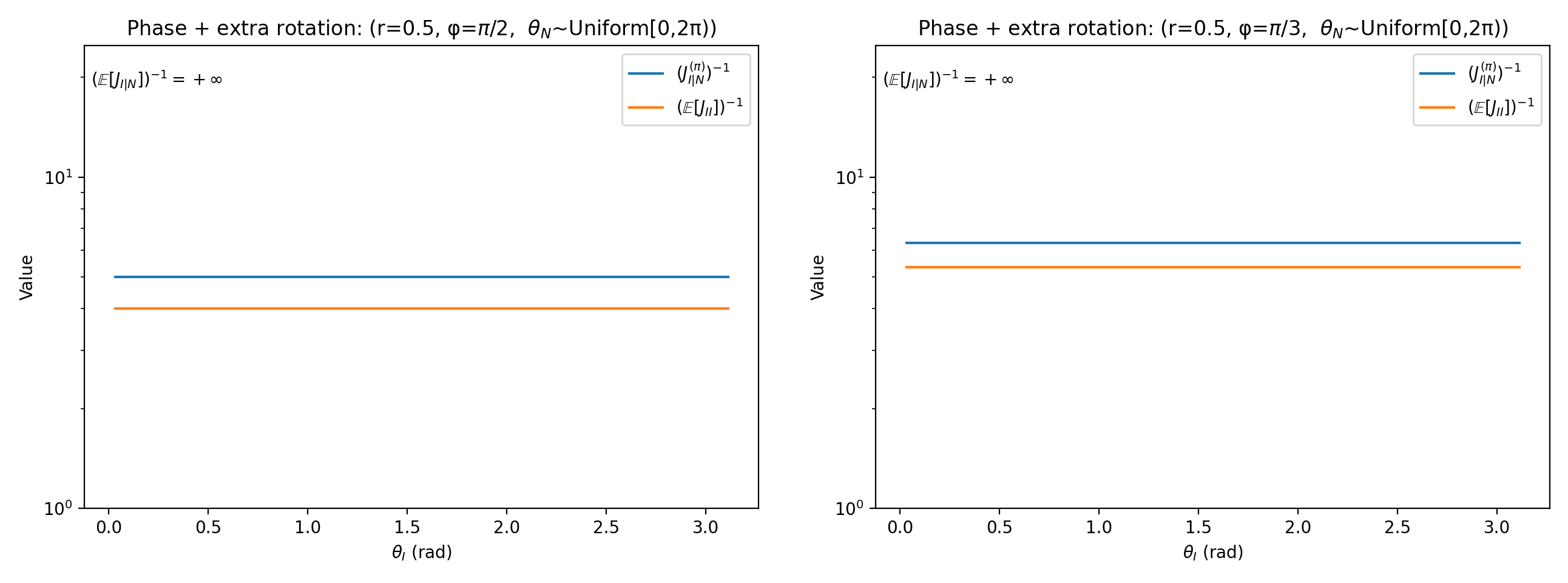}
    \vspace{-3mm}
    \includegraphics[width=1.0\linewidth]{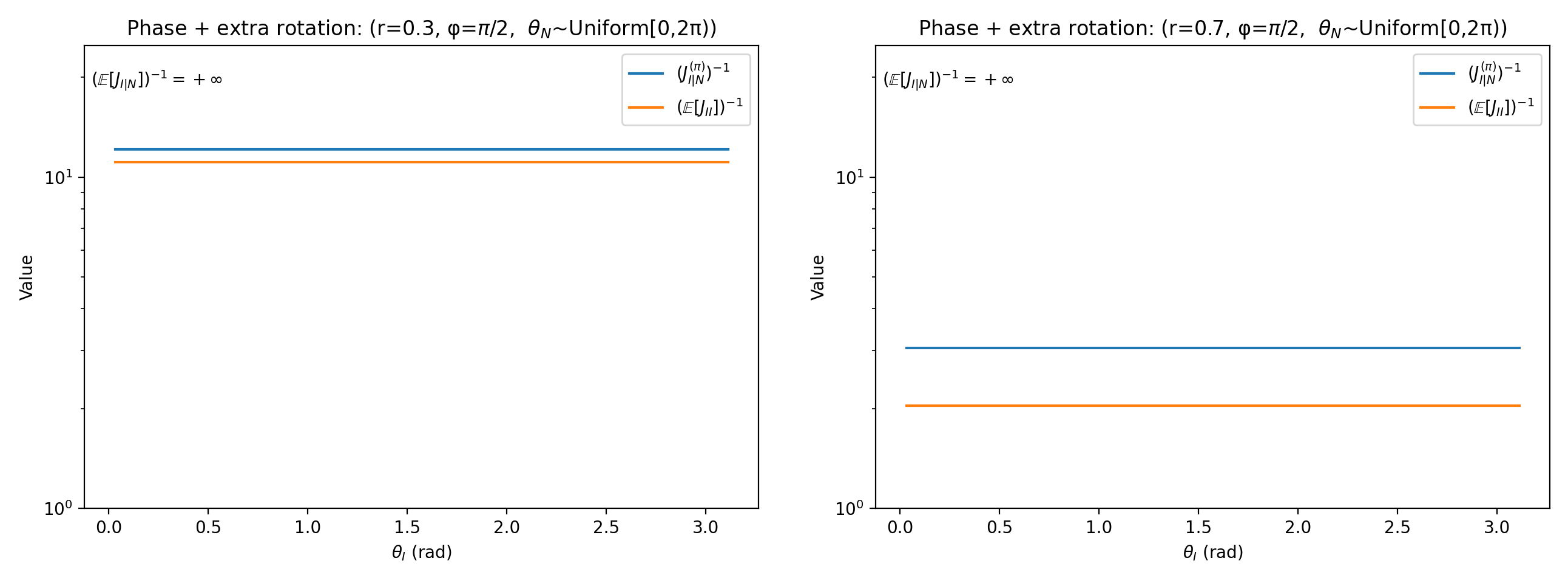}
    \caption{Comparison of the hpQFIM and its two approximations for the model~(phase with extra rotation).  
    The model parameters are set as $(r=0.5,~\phi=\pi/2)$, $(r=0.5,~\phi=\pi/3)$, $(r=0.3,~\phi=\pi/2)$ and
    $(r=0.7,~\phi=\pi/2)$ with uniform $\theta_N \sim U [0,2\pi)$.   
    Blue: $[\mathbb{E}(J_{I|N})]^{-1}$; Orange: $(J^{(\pi)}_{I|N})^{-1}$.}
    \label{fig:extrarotation}
\end{figure}

We now discuss consequences seen in the plots. 
To simplify the notation, denote the three quantities by
\[
U:=\big(\mathbb{E}[J_{I|N}]\big)^{-1},\qquad 
M:=\big(J^{(\pi)}_{I|N}\big)^{-1},\qquad 
L:=\big(\mathbb{E}[J_{II}]\big)^{-1}.
\]
The four panels (for $(r,\phi)\in
\{(0.5,\tfrac\pi2),\ (0.5,\tfrac\pi3),\ (0.3,\tfrac\pi2),\ (0.7,\tfrac\pi2) \}$) in Figure~\ref{fig:extrarotation} exhibit:

\begin{itemize}
  \item \textbf{Flat (constant) curves in $\theta_I$.} Since $J_{II},J_{NN},J_{IN}$ do not depend on angles, both $M$ and $L$ are constant in $\theta_I$. This matches the horizontal lines in all panels.
  \item \textbf{Divergent upper bound.} Because $J_{I|N}\equiv 0$, we have
  $\mathbb{E}[J_{I|N}]=0$ and therefore
  \[
  U=\big(\mathbb{E}[J_{I|N}]\big)^{-1}=+\infty,
  \]
  as indicated by the ``$U=+\infty$'' annotation. 
  \item \textbf{Lower and middle terms.} Averaging yields
  \[
  L=\frac{1}{\mathbb{E}[J_{II}]}=\frac{1}{r^2 \sin^2 \phi},\qquad
  J^{(\pi)}_{I|N}= \frac{r^2 \sin^2 \phi\,J_\pi}{r^2 \sin^2 \phi+J_\pi}
  \ \Longrightarrow\
  M=\frac{r^2 \sin^2+J_\pi}{r^2 \sin^2 \,J_\pi},
  \]
  where $J_\pi>0$ denotes the prior Fisher information for the nuisance. Thus $M>L$ and the hybrid inequality
  $U\ge M\ge L$ holds everywhere.
\end{itemize}

Next, we consider scaling with $(r,\phi)$. From the analytical expressions, we immediately see that only the scale $r^2\sin^2\phi$ matters:
\[
L=\frac{1}{r^2 \sin^2}\quad \text{and}\quad 
M=\frac{r^2 \sin^2+J_\pi}{r^2 \sin^2\,J_\pi}.
\]
Therefore increasing $r$ (with $\phi$ fixed) or increasing $\sin\phi$ (with $r$ fixed) uniformly lowers both constant lines. This is exactly what is observed when comparing
$r=0.3$ vs.\ $0.7$ at $\phi=\pi/2$, and $\phi=\pi/2$ vs.\ $\pi/3$ at $r=0.5$.

To summarize the result of this model, we conclude as follows. 
The extra-rotation coupling makes the score directions for $\theta_I$ and $\theta_N$ collinear, so the likelihood-only Schur complement $J_{I|N}$ vanishes identically and pure likelihood
information on $\theta_I$ is lost. Introducing a nonzero prior concentration $J_\pi$ on the nuisance
regularizes this degeneracy: the hybrid information becomes
$J^{(\pi)}_{I|N}=r^2\sin^2\phi J_\pi/(r^2\sin^2\phi+J_\pi)$, yielding the finite middle curve
$M=(J^{(\pi)}_{I|N})^{-1}=(r^2\sin^2\phi +J_\pi)/(r^2\sin^2\phi J_\pi)$, while the lower bound
$L=(\mathbb{E}[J_{II}])^{-1}=1/r^2\sin^2\phi$ represents the naive baseline. The observed flatness of $M$ and $L$ in $\theta_I$ reflects that, for this model, only the global scale $r^2\sin^2\phi$ and the prior concentration $J_\pi$ determine estimation precision.

\subsection{Additional sine model}
For the model
\[
\mathbf{s}(\theta_I,\theta_N\mid r,\phi)
=\big(r\sin\phi\cos(\theta_I),\ r\sin\phi\sin(\theta_I+\theta_N),\ r\cos\phi\big),
\]
with interest $\theta_I$ and nuisance $\theta_N\sim\mathrm{Unif}(0,1]$. 
The result is shown in Figure~\ref{fig:additionalsine}.
\begin{figure}[!htbp]
    \centering
    \includegraphics[width=1.0\textwidth]{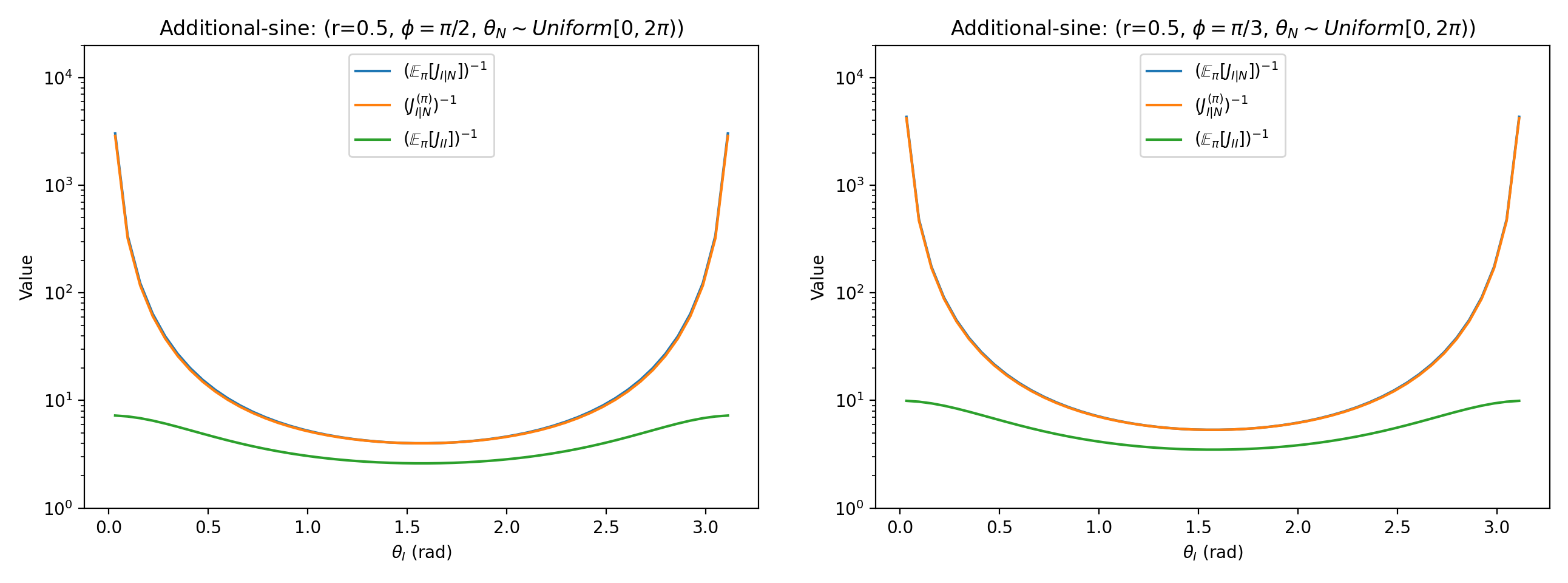}
    \includegraphics[width=1.0\textwidth]{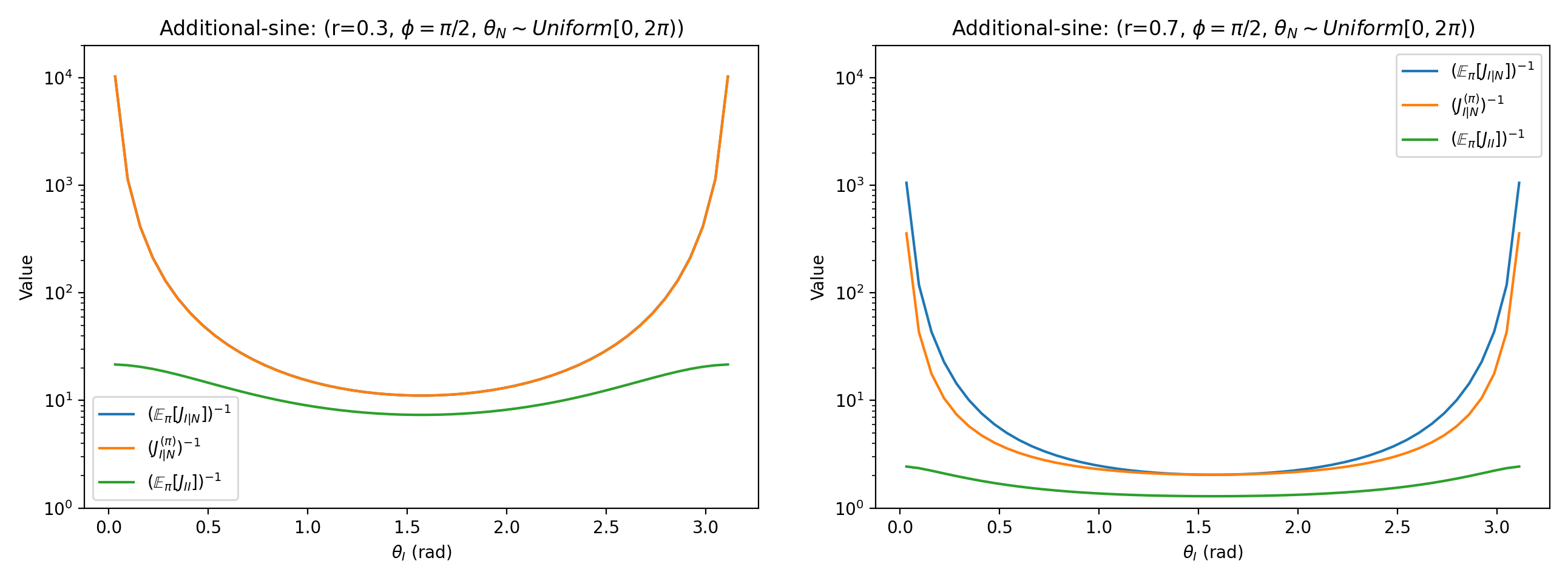}
    \caption{Comparison of the hpQFIM and its two approximations for the model~(phase with extra rotation).  
    The model parameters are set as $(r=0.5,~\phi=\pi/2)$, $(r=0.5,~\phi=\pi/3)$, $(r=0.3,~\phi=\pi/2)$ and $(r=0.7,~\phi=\pi/2)$ with uniform $\theta_N \sim U [0,2\pi)$.   
    Blue: $[\mathbb{E}(J_{I|N})]^{-1}$; Orange: $(J^{(\pi)}_{I|N})^{-1}$; Green: $(\mathbb{E}[J_{II}])^{-1}$.}
    \label{fig:additionalsine}
\end{figure}

We discuss consequences of the plots.
Denote, as before,
\[
U:=\big(\mathbb{E}_\pi[J_{I\mid N}]\big)^{-1},\qquad 
M:=\big(J^{(\pi)}_{I\mid N}\big)^{-1},\qquad 
L:=\big(\mathbb{E}_\pi[J_{II}]\big)^{-1}.
\]

First, we see that across all parameter choices, the ordering $U \ge M \ge L$ holds for every $\theta_I$. 
The curves are symmetric about $\theta_I=\pi/2$ and reach their minima near the center of the interval. 
As $\theta_I \to 0$ or $\pi$, $M$ and $L$ grow rapidly, reflecting the near-colinearity of score directions for $\theta_I$ and $\theta_N$ at the endpoints. 
In the bulk of the interval, $U$ and $M$ are nearly indistinguishable under the uniform prior, showing that averaging largely cancels cross-term fluctuations. 

Next, we analyze the dependence on the model parameters $(r,\phi)$ .
The overall information scale $r^2\sin^2\phi$ governs the depth and position of the curves. 
Increasing $r$ at fixed $\phi$ uniformly lowers all bounds and makes the central valley deeper. 
Reducing $\sin\phi$ at fixed $r$ (e.g.\ $\phi:\ \tfrac{\pi}{2}\to\tfrac{\pi}{3}$) decreases $r^2\sin^2\phi$ and shifts the curves upward while preserving their shapes. 
These monotone trends are consistently observed across the four chosen parameter settings.

We hence summarize the second model as follows. 
The additional-sine model highlights how an asymmetric coupling of the nuisance to the signal
modifies estimation: the nuisance phase enters only the second transverse component,
so the two parameter directions do not affect the state in a rotationally symmetric way. 
With a uniform prior over the nuisance, the lower bound $L=(\mathbb{E}_\pi[J_{II}])^{-1}$ is a proxy in the bulk of $\theta_I$, 
but it becomes optimistic in a narrow neighborhood of the endpoints ($\theta_I\!\to\!0,\pi$), 
where the middle curve $M=(J^{(\pi)}_{I|N})^{-1}$ better reflects the true loss of information. 
Geometrically, this behavior follows from the anisotropic suppression of the Schur complement:
as the score vectors $\partial_{\theta_I}\mathbf s$ and $\partial_{\theta_N}\mathbf s$ become nearly colinear near the endpoints, 
the conditional information $J_{I|N}$ is reduced more strongly than $J_{II}$. 
Physically, the model abstracts scenarios where an unwanted phase acts only on a single quadrature/polarization, 
breaking rotational symmetry; this selective action makes it a minimal yet expressive setting to study the interplay between interest and nuisance.

\subsection{Anisotropic shrinking model}
We consider the anisotropic shrinking model
\[
\mathbf{s}(\theta_I,\theta_N\mid r,\phi)
=\big(r\sin\phi\cos\theta_I,\; r\,\theta_N\sin\phi\sin\theta_I,\; r\cos\phi\big),
\]
with interest $\theta_I$ and nuisance $\theta_N\sim\mathrm{Unif}(0,1]$. 
The result is shown in Figure~\ref{fig:anisotropic}. 
\begin{figure}[!htbp]
\centering
\includegraphics[width=1.0\linewidth]{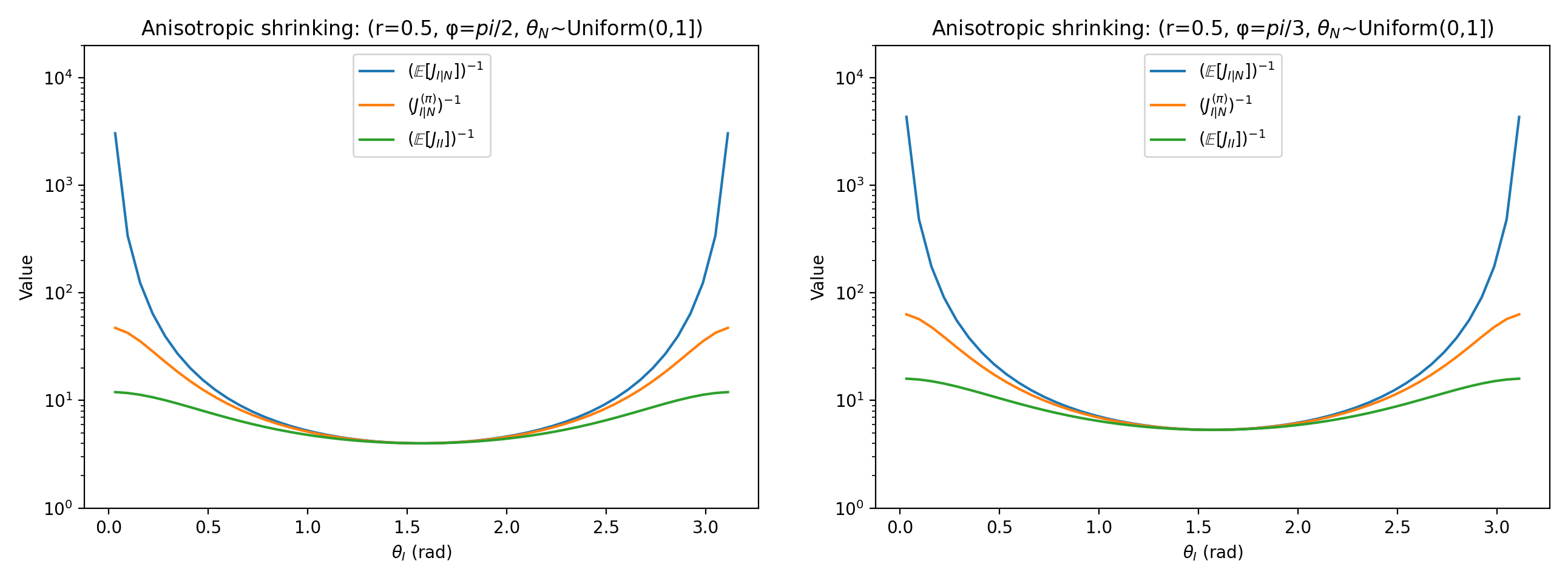}
\includegraphics[width=1.0\linewidth]{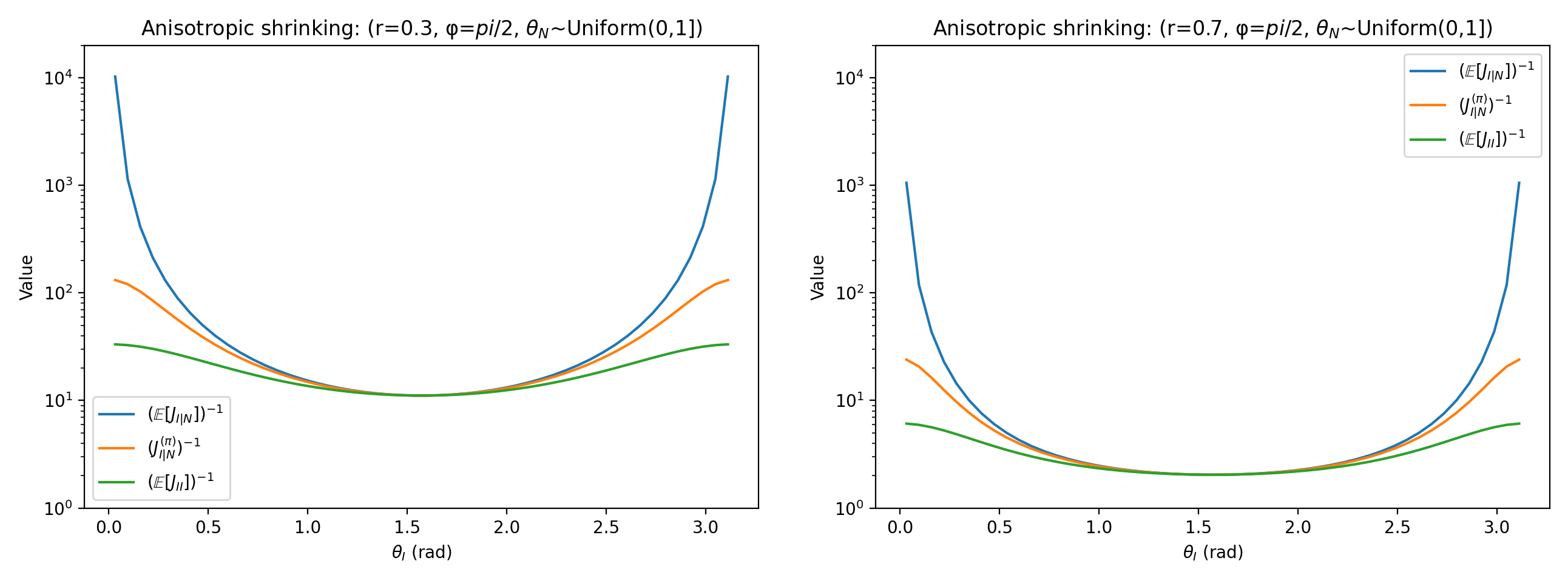}
\caption{Comparison of the hpQFIM and its two approximations for the model~(anisotropic shrinking) The model parameters are set as $(r=0.5, ~\phi=\pi/2)$, $(r=0.5,~ \phi=\pi/3)$, $(r=0.3,~\phi=\pi/2)$ and $(r=0.7,~\phi=\pi/2)$ with uniform $\theta_N \sim U [0,1]$. Blue: $[\mathbb{E}(J_{I|N})]^{-1}$; Orange: $(J^{(\pi)}_{I|N})^{-1}$; Green: $(\mathbb{E}[J_{II}])^{-1}$.
    }
\label{fig:anisotropic}
\end{figure}

We discuss the properties of this model. 
Denote, as before,
\[
U:=\big(\mathbb{E}_\pi[J_{I\mid N}]\big)^{-1},\qquad 
M:=\big(J^{(\pi)}_{I\mid N}\big)^{-1},\qquad 
L:=\big(\mathbb{E}_\pi[J_{II}]\big)^{-1}.
\]

First, we look at ordering and endpoint behavior. 
Across all panels the hybrid inequality $U\ge M\ge L$ holds point-wise in $\theta_I$. With a uniform prior over $\theta_N$, the plots are approximately symmetric about $\theta_I=\pi/2$ and attain their minima near the center.  
This symmetry follows from the trigonometric structure of the model, in which only $\sin\theta_I$ and $\cos\theta_I$ enter the QFIM blocks.
As $\theta_I\to 0$ or $\pi$, both $U$ and $M$ rise rapidly, while $L$ increases only moderately.  
This reflects a geometric degeneracy at the endpoints: $\partial_{\theta_I}\mathbf s$ and $\partial_{\theta_N}\mathbf s$ both carry a factor $\sin\theta_I$, so their norms---and hence the Schur complement $J_{I\mid N}$---are strongly suppressed there.  
Averaging over $\theta_N$ therefore drives $\mathbb{E}_\pi[J_{I\mid N}]$ to small values, inflating $U$ (and, to a lesser extent, $M$), whereas $J_{II}$ alone does not vanish at the same rate, keeping $L$ comparatively low.  
Thus $L$ becomes loose in a narrow neighborhood of the endpoints, while $M$ better reflects the cost of eliminating the nuisance.

Next, tightness in the central region should be stressed.
Around the symmetric center $\theta_I\approx \pi/2$, the three curves approach one another and the gap $U-M$ becomes small; $M$ nearly coincides with $U$ and also approaches $L$.  
Away from the degeneracy, ``average of Schur'' and ``Schur of averages'' yield very similar information.

Last, we examine the dependence on the model parameters $(r,\phi)$. 
The panels for
\[
(r,\phi)\in\{(0.5,\tfrac{\pi}{2}),\ (0.5,\tfrac{\pi}{3}),\ (0.3,\tfrac{\pi}{2}),\ (0.7,\tfrac{\pi}{2})\}
\]
exhibit the expected following two trends. First, increasing $r$ at fixed $\phi$ (compare $r=0.3$ vs.\ $0.7$ at $\phi=\pi/2$) increases Fisher information and hence lowers all three curves uniformly; the valley near $\theta_I=\pi/2$ deepens. Second, reducing $\sin\phi$ at fixed $r$ (e.g.,\ $\phi:\ \tfrac{\pi}{2}\to\tfrac{\pi}{3}$ for $r=0.5$) weakens the transverse component and raises all curves, with shapes essentially unchanged.

We summarize the behaviors of this model as follows. 
The anisotropic contraction along the nuisance-controlled axis amplifies endpoint degeneracy and creates a clear separation between $M$ and $L$ only where $\sin\theta_I$ is small; elsewhere $M$ remains close to $U$.  
Overall, the estimation precision is controlled chiefly by the scale $r\sin\phi$ and by avoiding the endpoint region, providing concrete guidance for operating points in hybrid estimation.

\section{Example (direction estimation)}
In this section we analyze a qubit model that admits closed-form expressions in the multiparameter setting, and then specialize it to a hybrid framework where a prior is placed only on radial nuisance parameter $r$, while the directional parameters $(\theta,\phi)$ are estimated in the frequentist sense. The benchmark is directly linked to entropic characteristics of the state: for a qubit with Bloch radius $r$, the spectrum is $\{(1\pm r)/2\}$ and the von Neumann entropy $S(\rho)$ is a monotone function of~$r$, making radius estimation operationally relevant for purity assessment.

\subsection{Example: Bloch-radius model with directional interest and radial nuisance}\label{subsec:bloch-model}
We consider single-qubit states in Bloch form
\begin{equation}\label{eq:bloch-state}
    \rho(r,\theta,\phi)
    = \tfrac{1}{2}\Big[
        I + r\big(\sin\theta\cos\phi\,\sigma_x
                 + \sin\theta\sin\phi\,\sigma_y
                 + \cos\theta\,\sigma_z\big)
      \Big],
\end{equation}
where the parameters of interest are the spherical angles $(\theta,\phi)\in[0,\pi]\times[0,2\pi)$ and the nuisance parameter is the Bloch radius $r\in(0,1)$. Let $\mathbf{s}(r,\theta,\phi)=r\,\hat{\mathbf n}(\theta,\phi)$ with
\[
\hat{\mathbf n}(\theta,\phi)=(\sin\theta\cos\phi,\ \sin\theta\sin\phi,\ \cos\theta).
\]
The parameter derivatives of the Bloch vector are
\[
\partial_r\mathbf s=\hat{\mathbf n},\qquad
\partial_\theta\mathbf s=r\,\partial_\theta\hat{\mathbf n},\qquad
\partial_\phi\mathbf s=r\,\partial_\phi\hat{\mathbf n},
\]
and the elementary spherical identities
\[
\hat{\mathbf n}\!\cdot\!\partial_\theta\hat{\mathbf n}
=\hat{\mathbf n}\!\cdot\!\partial_\phi\hat{\mathbf n}=0,\quad
\|\partial_\theta\hat{\mathbf n}\|=1,\quad
\|\partial_\phi\hat{\mathbf n}\|=\sin\theta,\quad
\det[\partial_\theta\hat{\mathbf n},\hat{\mathbf n},\partial_\phi\hat{\mathbf n}]=-\sin\theta ,
\]
will be used below.

\paragraph{SLD quantum Fisher information.}
The QFIM for the parameter order $(r,\theta,\phi)$, i.e., $\theta_1=r,\theta_2=\theta,\theta_3=\phi$, is
\[
    J(r,\theta,\phi)
    = \operatorname{diag}\!\big((1-r^2)^{-1},\, r^2,\, r^2\sin^2\theta\big).
\]
We thus have $J_{r\theta}=J_{r\phi}=J_{\theta\phi}=0$.
Recall that the parameters of interest are $\theta_I=(\theta,\phi)$, whereas the nuisance parameter is $\theta_N=r$. 

\paragraph{Hybrid partial information and CR-type bound.}
Treating $(\theta,\phi)$ as interest and $r$ as nuisance, the partial SLD information (conditioning on $r$) is the Schur complement
\[
    J_{I|N}
    = J_{II} - J_{IN} \big(J_{NN}\big)^{-1} J_{NI}.
\]
Since the cross terms vanish, we have
\begin{equation}\label{eq:oracle-info-sld}
    J_{I|N}(\theta,\phi;\,r)=J_{II}(\theta,\phi;\,r)
    = r^2\begin{pmatrix}1&0\\[2pt]0&\sin^2\theta\end{pmatrix}.
\end{equation}
Consequently, for any prior $\pi_r$ on $r$ the three quantities of Theorem~\ref{thm:hybridlbub} are identical, i.e.,
\begin{align}
(\mathbb E_{\pi_r}[J_{I|N}(\theta,\phi)])^{-1}=
\big(J^{(\pi_r)}_{I|N}(\theta,\phi)\big)^{-1}
= \big(\mathbb E_{\pi_r}[J_{II}(\theta,\phi)]\big)^{-1}
= \frac{1}{\mathbb E_{\pi_r}[r^2]}\begin{pmatrix}1&0\\[2pt]0&\sin^{-2}\theta\end{pmatrix}.
\end{align}
Thus, in point estimation (oracle $r$ known) the attainable variance for unbiased estimation of $(\theta,\phi)$ depends on the value of $r$ via \eqref{eq:oracle-info-sld}; in the hybrid setting it depends on the prior $\pi_r$ through $\mathbb E_{\pi_r}[r^2]$. Therefore, one can certify a direction-estimation risk lower bound via the prior without knowing the true $r$.

\section{Conclusions}
In this paper, we proposed a hybrid estimation framework that treats parameters of interest and nuisance parameters asymmetrically by placing a prior only on the nuisance sector while estimating the parameter of interest in the frequentist sense. Within this framework we (i) defined the hybrid partial quantum Fisher information matrix (hpQFIM) by prior-averaging the nuisance block and taking the Schur complement on the interest block; (ii) derived the associated hybrid Cram\'er-Rao-type (CR-type) bound; (iii) clarified the operational gain over pure point estimation---hybrid-optimal measurements depend on the prior over the nuisance rather than its unknown true value; and (iv) established inequalities that relate prior-averaged quantities and elucidate their limiting behaviors.

To make the discussion concrete, we analyzed an analytically solvable single-qubit model where the direction $(\theta,\phi)$ is the parameter of interest and the Bloch radius $r$ plays the role of nuisance (with a prior $\pi_r$). Because the QFIM is diagonal (cross terms vanish), the hybrid partial information for $(\theta,\phi)$ reduces to the prior average of the interest block, yielding the CR-type matrix bound
\[
\big(J^{(\pi_r)}_{I\mid N}\big)^{-1}
= \frac{1}{\mathbb{E}_{\pi_r}[r^2]}\,\operatorname{diag}\!\big(1,\sin^{-2}\theta\big),
\]
while with oracle knowledge of $r$ the bound scales as $1/r^2$; under a genuine prior the hybrid quantities coincide.

Beyond this solvable case, our inequalities position the hybrid bound between natural Bayesian and point-estimation surrogates, thereby quantifying how prior knowledge about nuisance parameters can be systematically leveraged without fully committing to a Bayesian treatment of all parameters. These results give a unified and operationally transparent picture of nuisance handling in quantum metrology.

Several directions follow naturally. We list three possible extensions of this work. 
\begin{itemize}
  \item \textbf{Tightness and achievability.} Determining conditions under which the hybrid CR-type bound is tight and characterizing the structure of achieving measurements- especially beyond the qubit radius example- remain open. Connections to D-invariant models and to measurement classes with symmetry constraints are promising~\cite{suzuki2019information}.
  \item \textbf{Prior modeling and robustness.} Moving from uniform priors to anisotropic families such as von Mises--Fisher priors enables a controlled interpolation between ignorance and alignment. Quantifying robustness of hybrid-optimal measurements against prior misspecification is an important practical question.
  \item \textbf{Full hybrid model.} In this work we focus on the canonical hybrid setting where the nuisance parameters are random (with a prior) and the parameters of interest are nonrandom (point-wise). A natural next step is the full hybrid model, in which parameters are partitioned into four classes: interest-random, interest-nonrandom, nuisance-random, and nuisance-nonrandom.

\end{itemize}

Overall, the hybrid viewpoint separates what must be learned (the parameter of interest) from what can be integrated out using prior structure (the nuisance), yielding bounds and design principles that are both rigorous and operationally meaningful. We expect this perspective to be broadly useful for quantum metrology in low-copy and resource-constrained scenarios where nuisance is inevitable.

\vspace{6pt} 


\section*{Data availability}
The data and program code that support the findings of this study are available from the corresponding author upon reasonable request.

\section*{Acknowledgments}
The work is partly supported by JSPS KAKENHI Grant Number JP24K14816, and ERATO “Super Quantum Entanglement” (Grant No. JPMJER2402) from JST. J.Z. is also supported by the research assistant scholarship at the University of Electro-Communications. Thanks to the suggestion from Dr.~Koichi Yamagata.

\appendix
\section[\appendixname~\thesection]{Quantum hybrid CR-type lower bound}
\label{app:Quanhybridcrlowerbound}
This part is the proof for quantum hybrid CR-type lower bound which means for any POVM $\Pi$ and any locally unbiased estimator $\hat\theta_I$,
\begin{align}
V_{\theta_I,\pi}(\Pi,\hat\theta_I)\ \succeq\ \big(J^{(\pi)}_{I\mid N}(\theta_I)\big)^{-1}.
\end{align}
This can be proved by the same process of van Trees inequality \cite{van2004detection}. \\
\textbf{Proof strategy:}
We first construct a prior-augmented information matrix, compare quantum and classical information under the prior, and then we consider the matrix covariance inequality. The desired $II$-block bound follows from a block-inverse (Schur complement) identity.

\noindent\textbf{Regularity.} We use standard conditions (dominated convergence/Leibniz rule, fixed support in $x$, and a $\theta_I$–independent prior) to interchange $\partial_{\theta_I}$ with integrals.

\noindent\textbf{Prior-averaged quantum information matrix:} The following matrix collects the Fisher information averaged over the nuisance prior and add the prior Fisher $J_\pi$ on the nuisance block; this is the natural hybrid analogue of the van Trees setup. Define the prior-averaged quantum information matrix
\begin{align*}
G^{(\pi)}(\theta_I) :=\ \mathbb{E}_{\pi}\!\big[J(\theta_I,\theta_N)\big]\;+ \begin{pmatrix}
    0 & 0 \\ 0 & J_\pi
\end{pmatrix},
\end{align*}
and the corresponding classical information matrix 
\begin{align*}
G^{(\pi)}(\theta_I\,|\,\Pi) :=\ \mathbb{E}_{\pi}\!\big[J(\theta_I,\theta_N\,|\,\Pi)\big]\;+ \begin{pmatrix}
    0 & 0 \\ 0 & J_\pi
\end{pmatrix}.
\end{align*}
By the construction of quantum Fisher information matrix, we the following inequality with classical Fisher information matrix \cite{braunstein1994statistical}
\begin{align}
    J(\theta_I,\theta_N) \succeq J(\theta_I,\theta_N\,|\,\Pi),
\end{align}
with 
\begin{align*}
    J(\theta_I , \theta_N \, | \, \Pi)_{ij} = \mathbb E_{x|\theta_I, \theta_N} \left[ \frac{\partial \log p(x|\theta_I,\theta_N)}{\partial \theta_{I,i}} \frac{\partial \log p(x|\theta_I,\theta_N)}{\partial \theta_{I,j}} \right].
\end{align*}
Thus, 
\begin{align*}
    G^{(\pi)}(\theta_I) \succeq G^{(\pi)}(\theta_I\,|\,\Pi)\text{ and also }
    G^{(\pi)}(\theta_I\,|\,\Pi)^{-1} \succeq G^{(\pi)}(\theta_I)^{-1}.
\end{align*}
The next step is the core one of which the details are shown later. We notice
\begin{align}\label{eq:corestep}
    V_{\theta_I,\pi}(\Pi,\hat\theta_I)\ \succeq \left( G^{(\pi)}(\theta_I\,|\,\Pi)^{-1}\right)_{II}.
\end{align}
Therefore, we have the theorem,
\begin{align*}
    V_{\theta_I,\pi}(\Pi,\hat\theta_I)\ \succeq \left( G^{(\pi)}(\theta_I)^{-1}\right)_{II}=\;\big(J^{(\pi)}_{I\mid N}(\theta_I)\big)^{-1}.
\end{align*}
The remaining of this part is to prove the inequality \ref{eq:corestep}. Let the vectors of the functions of random variables be
\begin{align*}
    &f(x)=[ f_1(x)~ f_2(x)]^\top,\quad f_1(x)=\hat \theta_I(x)-\theta_I,\quad f_2(x)=\hat \theta_N(x)-\theta_N, \\
    &g(x)=[ g_1(x)~ g_2(x)]^\top,\quad g_1(x) = [\partial_i \log p(x|\theta_I,\theta_N)],\quad\\
    & g_2(x) = [\partial_a \log \pi(\theta_N) p(x|\theta_I,\theta_N)],
\end{align*}
where $\partial_i =\frac{\partial}{\partial \theta_{I,i}}$, $\partial_a =\frac{\partial}{\partial \theta_{N,a}}$. Let \(\mathbb E[\cdot]\) be the total expectation, 
\begin{align*}
    \mathbb E[\cdot] = \int dx \int d\theta_N \,\pi(\theta_N)\, p(x|\theta_I,\theta_N) \cdot.
\end{align*} Then by the covariance inequality \cite{van2004detection} which is 
\begin{align*}
     F \succeq T G^{-1} T^\top,
\end{align*}
where $F=\mathbb E[f f^\top]$, $T=\mathbb E[f g^\top]$, $G=\mathbb E[g g^\top]$, we obtain
\begin{align*}
    F=\mathbb E[f f^\top]=\mathbb E \left[  \begin{pmatrix}
        f_1 \\ f_2
    \end{pmatrix} \begin{pmatrix}
        f_1 & f_2
    \end{pmatrix}\right] 
    = \mathbb E \left[ \begin{pmatrix}
        (\hat \theta_I-\theta_I)(\hat \theta_I-\theta_I)^\top &
        (\hat \theta_I-\theta_I)(\hat \theta_N-\theta_N)^\top \\
        (\hat \theta_N-\theta_N)(\hat \theta_I-\theta_I)^\top &
        (\hat \theta_N-\theta_N)(\hat \theta_N-\theta_N)^\top
    \end{pmatrix} \right].
\end{align*}
For the parameter vector $(\theta_I,\theta_N)$, $F$ is the MSE of it.\\  Claim: \(T\) is the identity matrix. Proof:
\begin{align*}
    T=\mathbb E[fg^\top] &= \mathbb E \left[ \begin{pmatrix}
        \hat \theta_I-\theta_I \\ \hat \theta_N - \theta_N
    \end{pmatrix} \begin{pmatrix}
        \partial_i^\top \log p(x|\theta_I,\theta_N)&
        \partial_a^\top \log (\pi(\theta_N) p(x|\theta_I,\theta_N)) 
    \end{pmatrix} \right] \\
    &= \mathbb E \left[ \begin{pmatrix}
        (\hat \theta_I-\theta_I) \nabla_{\theta_I}^\top & (\hat \theta_I-\theta_I) \nabla_{\theta_N}^\top\\
        (\hat \theta_N-\theta_N) \nabla_{\theta_I}^\top & (\hat \theta_N-\theta_N) \nabla_{\theta_N}^\top
    \end{pmatrix} \right],
\end{align*}
where $\nabla_{\theta_I} = [\partial_i \log p(x|\theta_I,\theta_N)]$, $\nabla_{\theta_N}= [\partial_a \log( \pi(\theta_N) p(x|\theta_I,\theta_N) )]$. By the result of the van Trees inequality, we know $(\hat \theta_N-\theta_N) \nabla_{\theta_N}^\top = I_{\theta_N}$, $I_{\theta_N}$ means the identity matrix in the dimension of $\theta_N$. We need to check it one by one since the expectations are different. 
\begin{align*}
    \mathbb E[(\hat \theta_I-\theta_I)\nabla_{\theta_I}^\top]_{ij} &= \int \int dx \, d\theta_N \,\pi(\theta_N)\, p(x|\theta_I,\theta_N)( \hat \theta_{I,i}(x) -\theta_{I,i}) \partial_j \log p(x|\theta_I,\theta_N)\\ 
    &= \int \int dx \, d\theta_N \,\pi(\theta_N)\, ( \hat \theta_{I,i}(x) -\theta_{I,i}) \partial_j p(x|\theta_I,\theta_N)\\  
    &= \int \int dx \, d\theta_N \, ( \hat \theta_{I,i}(x) -\theta_{I,i}) \partial_j p(x,\theta_N |\theta_I)\\  
    &= \int dx \, ( \hat \theta_{I,i}(x) -\theta_{I,i}) \partial_j p(x|\theta_I)\\
    &= \partial_j \int dx \, \hat \theta_{I,i} (x) p(x|\theta_I) = \delta_{i,j},
\end{align*}
which is given by interchanging the derivative and integral. Then we need to show the remaining two off-diagonal items are zero matrix.
\begin{align*}
    \mathbb E[(\hat \theta_N-\theta_N) \nabla_{\theta_I}^\top]_{ai} &= \!\int \!\int\!  dx  d\theta_N \,\pi(\theta_N)\, p(x|\theta_I,\theta_N) (\hat \theta_{N,a}(x)-\theta_{N,a}) \partial_i \log p(x|\theta_I,\theta_N) \\
    &=\int \int dx \, d\theta_N \,\pi(\theta_N)\, (\hat \theta_{N,a}(x)-\theta_{N,a}) \partial_i p(x|\theta_I,\theta_N)\\
    &=\int \int dx \, d\theta_N \,(\hat \theta_{N,a}(x)-\theta_{N,a}) \partial_i p(x,\theta_N|\theta_I)\\
    &=\int dx \, \hat \theta_{N,a} \,\partial_i p(x|\theta_I) -  \int d\theta_N \, \theta_{N,a} \,\partial_i p(\theta_N|\theta_I)\\
    &= \partial_i \theta_{N,a} - \partial_i \int d \theta_N \, \theta_{N,a} p(\theta_N|\theta_I) \\
    &= 0-0=0,
\end{align*}
where we used the independence between \(\theta_I\) and \( \theta_N \).
\begin{align*}
    \mathbb E[(\hat \theta_I-\theta_I) \nabla_{\theta_N}^\top]_{ia} 
    &= \scalebox{1}{$\displaystyle
      \int\!\!\int dx\, d\theta_N\, \pi(\theta_N)\, p(x|\theta_I,\theta_N)\,
      (\hat \theta_{I,i}(x)-\theta_{I,i})\,
      \partial_a\log\!\big(\pi(\theta_N)\, p(x|\theta_I,\theta_N)\big)
    $}\\
    &= \int \int dx \, d\theta_N  \,(\hat \theta_{I,i}(x)-\theta_{I,i}) \partial_a p(x,\theta_N|\theta_I)\\
    &= \int  dx  \,(\hat \theta_{I,i}(x)-\theta_{I,i}) \partial_a p(x|\theta_I) = 0,
\end{align*}
which is similar to the previous one. Thus we obtain that $T$ is the identity matrix. Then we factorize $G$ in
\begin{align*}
    G= \mathbb E \left[ \begin{pmatrix}
        \nabla_{\theta_I} \nabla_{\theta_I}^\top & \nabla_{\theta_I} \nabla_{\theta_N}^\top\\
        \nabla_{\theta_N} \nabla_{\theta_I}^\top & \nabla_{\theta_N} \nabla_{\theta_N}^\top
    \end{pmatrix} \right],
\end{align*}
and calculate the item one by one. Firstly, we get the form of $\mathbb E[\nabla_{\theta_N} \nabla_{\theta_N}^\top]=J_\pi + \mathbb E_\pi[J_{\theta_N}]$ by
\begin{align*}
    & \quad \mathbb E[\nabla_{\theta_N} \nabla_{\theta_N}^\top]_{ab}  \\
    &= \int \int \partial_a \left(\log \pi(\theta_N) p(x|\theta_I,\theta_N) \right) \partial_b \left(\log \pi(\theta_N) p(x|\theta_I,\theta_N) \right) p(x|\theta_I,\theta_N) \pi(\theta_N) \,d\theta_N \,dx\\
    &= \int \int \partial_a \left( \ell(\theta_N) + \ell(x|\theta_I,\theta_N)  \right) \partial_b \left( \ell(\theta_N) +\ell(x|\theta_I,\theta_N) \right) p(x|\theta_I,\theta_N) \pi(\theta_N) \,d\theta_N \,dx\\
    &= \displaystyle \int \partial_a \ell(\theta_N) \partial_b \ell(\theta_N) \pi(\theta_N) d\theta_N \\
    &\quad + \int \! \int \left( \partial_a \ell(\theta_N) \partial_b \ell(x|\theta_I,\theta_N) + \partial_a \ell(x|\theta_I,\theta_N)  \partial_b \ell(\theta_N)\right) p(x|\theta_I,\theta_N) \pi(\theta_N) \,d\theta_N \,dx \\
    &\quad+ \int \int \partial_a \ell(x|\theta_I,\theta_N)  \partial_b \ell(x|\theta_I,\theta_N) p(x|\theta_I,\theta_N) \pi(\theta_N) \,d\theta_N \,dx \\
    &= J_\pi + 0 + \int J_{\theta_N} \,\pi(\theta_N) \,d\theta_N \\
    &= J_\pi + \mathbb E_\pi [J_{\theta_N}],
\end{align*}
where 
\begin{align*}
    J_\pi &= \int \partial_a \ell(\theta_N) \partial_b \ell(\theta_N) \pi(\theta_N) \,d\theta_N,\quad\\
    J_{\theta_N} &= \int  \partial_a \ell(x|\theta_I,\theta_N)  \partial_b \ell(x|\theta_I,\theta_N) p(x|\theta_I,\theta_N) \,dx = \mathbb E_{x|\theta_I,\theta_N}\!\left[\partial_a \ell(x|\theta_I,\theta_N) \partial_b \ell(x|\theta_I,\theta_N) \right] ,\\ \mathbb E_{x|\theta_I,\theta_N}[\cdot] &= \int \cdot \, p(x|\theta_I,\theta_N)\,dx,\ \ \mathbb E_\pi[\cdot] = \int \cdot \, \pi(\theta_N)\,d\theta_N,\\
    \ell(\theta_N)&= \log \pi(\theta_N),~ \ell(x|\theta_I, \theta_N) = \log p(x|\theta_I,\theta_N).
\end{align*}
To show the middle part of the $ab$-th item is zero,
\begin{align*}
    \int \int \partial_a \ell(\theta_N) \partial_b \ell(x|\theta_I,\theta_N)  p(x|\theta_I,\theta_N) \pi(\theta_N) \,d\theta_N \,dx &= \int \int \partial_a \pi(\theta_N) \partial_b p(x|\theta_I,\theta_N) \,d\theta_N \,dx \\
    &= \int \partial_a \pi(\theta_N) \,d\theta_N \, \partial_b \int p(x|\theta_I,\theta_N)  \,dx \\
    &=\int \partial_a \pi(\theta_N) \,d\theta_N \, \partial_b (1) =0.
\end{align*}
The item of the parameter of interest $\theta_I$ is 
\begin{align*}
    \mathbb E[\nabla_{\theta_I} \nabla_{\theta_I}^\top]_{ij} &= \int \int \partial_i \log p(x|\theta_I,\theta_N) \partial_j \log p(x|\theta_I,\theta_N) p(x|\theta_I,\theta_N) \pi(\theta_N) \,d\theta_N \,dx\\
    &= \int \int \partial_i \ell(x|\theta_I,\theta_N) \partial_j \ell(x|\theta_I,\theta_N) p(x|\theta_I,\theta_N) \pi(\theta_N) \,d\theta_N \,dx\\
    &= \int J_{\theta_I} \,\pi(\theta_N) \,d\theta_N = \mathbb E_\pi [J_{\theta_I}],
\end{align*}
where $J_{\theta_I} = \int \partial_i \ell(x|\theta_I,\theta_N) \partial_j \ell(x|\theta_I,\theta_N) p(x|\theta_I,\theta_N) \,dx $.
The off-diagonal $i a$-th item is 
\begin{align*}
    \mathbb E[\nabla_{\theta_I} \nabla_{\theta_N}^\top]_{ia}
    &=\int \int dx \, d\theta_N \,\partial_i \log p(x|\theta_I,\theta_N) \partial_a \log ( p(x|\theta_I,\theta_N) \pi(\theta_N) )  \, \pi(\theta_N)\, p(x|\theta_I,\theta_N)\\
    &=\int \int dx \, d\theta_N \,\partial_i \ell(x|\theta_I,\theta_N) \big(\partial_a \ell(x|\theta_I,\theta_N) + \partial_a \ell(\theta_N) \big)  \, \pi(\theta_N)\, p(x|\theta_I,\theta_N) \\
    &= \int J_{(\theta_I,\theta_N)} \,\pi(\theta_N) \,d\theta_N +\int \int dx \, d\theta_N \,\partial_i \ell(x|\theta_I,\theta_N)  \partial_a \ell(\theta_N)  \, \pi(\theta_N)\, p(x|\theta_I,\theta_N) \\
    &= \mathbb E_{\theta_N}[J_{(\theta_I,\theta_N)}] + \int \int dx \, d\theta_N \,\partial_i p(x|\theta_I,\theta_N)  \partial_a \pi(\theta_N) \\
    &= \mathbb E_{\theta_N}[J_{(\theta_I,\theta_N)}] + \int d\theta_N \,\partial_a \pi(\theta_N) \,\partial_i (1) =  \mathbb E_{\theta_N}[J_{(\theta_I,\theta_N)}].
\end{align*}
As a result 
\begin{align*}
        G= \mathbb E \left[ \begin{pmatrix}
        \nabla_{\theta_I} \nabla_{\theta_I}^\top & \nabla_{\theta_I} \nabla_{\theta_N}^\top\\
        \nabla_{\theta_N} \nabla_{\theta_I}^\top & \nabla_{\theta_N} \nabla_{\theta_N}^\top
    \end{pmatrix} \right] = \mathbb E_\pi \left[ \begin{pmatrix}
        J_{(\theta_I,\theta_I)} & J_{(\theta_I,\theta_N)}\\
        J_{(\theta_N,\theta_I)} & J_{(\theta_N,\theta_N)}
    \end{pmatrix} \right] +
    \begin{pmatrix}
        0 & 0\\
        0 & J_\pi
    \end{pmatrix}.
\end{align*}
This implies our final result,
\begin{align*}
    V_{\theta_I, \theta_N, \pi}(\Pi,\hat \theta_I, \hat \theta_N) =  \mathbb E\begin{pmatrix}
        (\hat \theta_I-\theta_I)(\hat \theta_I-\theta_I)^\top &
        (\hat \theta_I-\theta_I)(\hat \theta_N-\theta_N)^\top \\
        (\hat \theta_N-\theta_N)(\hat \theta_I-\theta_I)^\top &
        (\hat \theta_N-\theta_N)(\hat \theta_N-\theta_N)^\top
    \end{pmatrix} \succeq G^{(\pi)} (\theta_I\, |\, \Pi)^{-1}.
\end{align*}
After taking the \( (I,I)\) entry, we obtain the demanding inequality,
\begin{align}
     V_{\theta_I, \pi}(\Pi,\hat \theta_I) \succeq \left( G^{(\pi)} (\theta_I\, |\, \Pi)^{-1}\right)_{II}.
\end{align}

\section[\appendixname~\thesection]{Lower and upper approximations for the hpQFIM}
\label{app:lowerandupper}
This part is the proof for Lower and upper approximiatons for the hpQFIM $J^{(\pi)}_{I|N}$. 
Note the following inequalities hold:
\begin{align}
\mathbb{E}_\pi[J_{II}(\theta_I,\theta_N)]
\ \succeq\ 
J^{(\pi)}_{I|N}(\theta_I)
\ \succeq\
\mathbb{E}_\pi \left[J_{I|N}(\theta_I,\theta_N)\right] .
\end{align}
Firstly, we proof the lower approximation such that
\begin{align*}
J^{(\pi)}_{I|N}(\theta_I) \succeq
\mathbb{E}_\pi \left[J_{I|N}(\theta_I,\theta_N)\right].
\end{align*}
By the Schur complement, we have that for any positive semidifinite \(m \times m\) matrix \(A\) and \( m \times n\) matrix \(B\), the following holds \cite{bhatia},
\begin{align*}
\begin{bmatrix} A & B \\ B^\dag & B^\dag A^{-1}B \end{bmatrix}\ \succeq\ 0,
\end{align*}
which means
\begin{align*}
\begin{bmatrix} \mathbb E[A] & \mathbb E[B] \\ \mathbb E[B^\dag] & \mathbb E[B^\dag A^{-1}B] \end{bmatrix}\ \succeq\ 0.    
\end{align*}
Since \(\mathbb E[A] \succeq 0\), this implies
\begin{align*}
\mathbb E_\pi[B^\dag A^{-1}B] \succeq \mathbb E_\pi[B^\dagger] \left( \mathbb E_\pi[A] \right)^{-1} \mathbb E_\pi[B].
\end{align*}
Substitute \(A= J_{NN},~ B=J_{NI}\) we have,
\begin{align*}
&\mathbb E_\pi[J_{IN} J_{NN}^{-1} J_{NI}] \succeq \mathbb E_\pi[J_{IN}] \left( \mathbb E_\pi[J_{NN}] \right)^{-1} \mathbb E_\pi[J_{NI}]    \\
\Longleftrightarrow &  \mathbb E_\pi[J_{II}] - \mathbb E_\pi[J_{IN} J_{NN}^{-1} J_{NI}] \preceq  \mathbb E_\pi[J_{II}] - \mathbb E_\pi[J_{IN}] \left( \mathbb E_\pi[J_{NN}] \right)^{-1} \mathbb E_\pi[J_{NI}] \\
\Longleftrightarrow &\mathbb{E}_\pi \left[J_{I|N}(\theta_I,\theta_N)\right]
\ \preceq \ 
J^{(\pi)}_{I|N}(\theta_I).
\end{align*}
Next, we prove the upper approximation such that 
\begin{align*}
J^{(\pi)}_{I|N}(\theta_I)
\ \preceq \ \mathbb{E}_\pi[J_{II}(\theta_I,\theta_N)].
\end{align*}
This is given by discarding the item in hpQFIM \( J^{(\pi)}_{I|N}(\theta_I) \),
\begin{align*}
    J^{(\pi)}_{I|N}(\theta_I) = \mathbb{E}_{\pi}\big[J_{II}\big]
\;-\;
\mathbb{E}_{\pi}\big[J_{IN}\big]\,
\Big(\,\mathbb{E}_{\pi}\big[J_{NN}\big]+J_\pi\Big)^{-1}
\mathbb{E}_{\pi}\big[J_{NI}\big] \preceq \mathbb E_\pi[J_{II}].
\end{align*}

\bibliographystyle{IEEEtranN}
\bibliography{myref}


%


\end{document}